\documentclass[12pt,a4paper]{article}
\usepackage{amsmath,amssymb,amsfonts,amsthm}
\newtheorem{theorem}{Theorem}[section]
\newtheorem{lemma}[theorem]{Lemma}

\newtheorem{corollary}[theorem]{Corollary}
\newtheorem{conjecture}[theorem]{Conjecture}
\usepackage{t1enc}
\usepackage[english]{babel}
\usepackage[latin1]{inputenc}
\usepackage{color}

\usepackage{graphicx}

\newcommand{\mf}{\mathcal{M}}

\title{Horizon area--angular momentum inequality for a class of axially
  symmetric black holes}

\author{Andr\'es Ace\~na$^{2}$ \and Sergio Dain$^{1,2}$ \and Mar\'\i a
  E. Gabach Cl\'ement$^{1}$\\ 
  \\
  $^1$Facultad de Matem\'atica, Astronom\'{i}a y F\'{i}sica, FaMAF,\\
  Universidad Nacional de C\'ordoba,\\
  Instituto de F\'{\i}sica Enrique Gaviola, IFEG, CONICET,\\
  Ciudad Universitaria, (5000) C\'ordoba, Argentina.  \\
  $^{2}$Max Planck Institute for Gravitational Physics,\\
  (Albert Einstein Institute), Am M\"uhlenberg 1,\\
  D-14476 Potsdam Germany. }

\begin{document}
\maketitle

\begin{abstract}
  We prove an inequality between horizon area and angular momentum for a class
  of axially symmetric black holes. This class includes initial conditions with
  an isometry which leaves fixed a two-surface. These initial conditions have
  been extensively used in the numerical evolution of rotating black
  holes. They can describe highly distorted black holes, not necessarily near
  equilibrium. We also prove the inequality on extreme throat initial data,
  extending previous results.
\end{abstract}

\section{Introduction}
\label{sec:introduction}
In a recent article \cite{dain10d} the following conjecture was formulated:
\begin{conjecture}
\label{c:1}
Consider an asymptotically flat, vacuum, complete axially symmetric initial
data set for Einstein equations. Then the following inequality holds
\begin{equation}
  \label{eq:1}
 8\pi |J|\leq  A,
\end{equation}
where $A$ and $J$ are the area and angular momentum of a connected component of
the apparent horizon.
\end{conjecture}
See \cite{dain10d} for a physical interpretation and motivation of the
inequality \eqref{eq:1}. In that article evidences for the validity of
conjecture \ref{c:1} were presented. These evidences are the following: in a
particular class of initial data set, called extreme throat initial data, the
first variation of the area, with fixed angular momentum, is zero and the
second variation is positive definite evaluated at the extreme Kerr throat
initial data. This indicates that the area of the extreme Kerr throat initial
data is a minimum among this class of data.  Since extreme Kerr initial data
satisfy the equality in (\ref{eq:1}) it follows that the area of generic throat
initial data satisfies (\ref{eq:1}). The key ingredient for this analysis is a
formula that relates the variations of the area of extreme throat initial data
with the variation of an appropriately defined mass functional.

However, as it was pointed out in \cite{dain10d}, in order to use these
arguments to prove conjecture \ref{c:1}, there are two main points that need to
be addressed. The first one is the following. It is well known that a
non-negative second variation is a necessary condition for a local minimum but
it is not sufficient.  To prove that extreme Kerr is a local minimum it is
necessary to provide extra estimates in a similar way as in \cite{Dain05d}.  As
remarked in \cite{dain10d} it is expected that the same analysis will
apply to this case also.  However, to prove that extreme Kerr is a global
minimum (which is, of course, what we need to prove) a different ingredient is
needed, since it is a priori not clear how to relate the area and the mass
functional mentioned above far from the extreme Kerr solution.

The second point is how to extend the result on extreme throat initial data to
include the physically relevant asymptotically flat black hole initial data
mentioned in the conjecture.  In \cite{dain10d} a limit procedure was proposed
which could in principle reduce the general case to the extreme throat
case. This limit procedure is similar in spirit to the extreme limit of the
Kerr black hole initial data. However, it is far from clear how to construct
this limit in general. A natural candidate would be initial data which are close
to Kerr. Even for that class of data the construction appears to be difficult.

The purpose of this article is to address the two points mentioned above. For
the first one we give a complete and optimal answer, namely, we prove that
extreme Kerr throat initial data are a global minimum in this class. At the
core of our argument lies a remarkable inequality that relates the area and
the mass functional for extreme throat initial data. This inequality is the
global generalization of the local arguments presented in \cite{dain10d}.

For the second point we give a partial answer. We prove the conjecture for a
class of initial data which has several technical restrictions. However,
despite of that, this class is relevant by itself.  It includes initial data
which have an isometry that leaves fixed a two-dimensional surface. This kind
of data can describe distorted rotating black holes far from equilibrium. They
have been extensively used in numerical simulations \cite{Brandt95}. The well
known Bowen-York family of initial data \cite{Bowen80} is a particular case,
where the metric is conformally flat. Our proof does not rely on a limit
procedure. Remarkable enough, it is a pure local argument, which uses the mass
formula to estimate in a simple fashion the area of the minimal surface.

Finally, we extend the validity of conjecture \ref{c:1} to include a
non-negative cosmological constant (and hence non-asymptotically flat initial
data). This generalization is relevant because there exists a counter-example of inequality (\ref{eq:1}) for
the case of negative cosmological constant, as it was pointed out in \cite{Booth:2007wu}. It will be seen that the inclusion of the cosmological
constant stresses the role of a non-negative Ricci scalar.

The plan of the article is the following. In section \ref{main} we present our
main result, given by theorem \ref{teobh}. We also discuss the scope of the
theorem and analyze relevant examples. The proof of this result consists of two
main parts, explained in sections \ref{secarea} and \ref{MassMinimum}.

The main result of section \ref{secarea} is an estimate for the area in terms
of the mass functional. In section \ref{MassMinimum} we present a variational
argument that asserts that the global minimum of the mass functional is given
by the extreme throat Kerr initial data.

\section{Main Result}\label{main}

An initial data set for Einstein equations, with cosmological constant
$\Lambda$, consists in a Riemannian 3-manifold $S$, together with its first and
second fundamental forms, $h_{ij}$ and $K_{ij}$ respectively, which satisfy the
vacuum Einstein constraints on $S$
\begin{align}
\label{hamiltonianConst}
R+K^2-K_{ij}K^{ij} &=2\Lambda,\\
\nabla^iK_{ij}-\nabla_jK &=0.\label{MomentumConst}
\end{align}
In these equations, $K=h^{ij}K_{ij}$, the Ricci scalar $R$, the contractions and
covariant derivatives are computed with respect to $h_{ij}$. The presence of
the cosmological constant $\Lambda$ allows for non asymptotically flat data
describing initial data of de Sitter ($\Lambda>0$) or anti-de Sitter
($\Lambda<0$) type.

When the initial data are maximal (i.e. $K=0$) the constraint equations
simplify considerably.  In particular, when $\Lambda\geq 0$, the scalar
curvature $R$ is non-negative. The condition $R\geq 0$ plays a crucial role in
this article.

The initial data are axially symmetric if there exists an axial Killing vector
field $\eta^i$ such that
\begin{equation}
\mathcal L_\eta h_{ij}=0,\qquad \mathcal L_\eta K_{ij}=0,
\end{equation}
where $\mathcal L$ denotes the Lie derivative. The Cauchy development of such
initial data will be an axially symmetric spacetime.

For an axially symmetric metric, it is always possible \cite{Chrusciel:2007dd}
to choose local coordinates ($r,\theta,\phi$) such that
\begin{equation}\label{metric}
h=e^ {\sigma}\left[e^ {2q}(dr^ 2+r^2d\theta^ 2)+r^2\sin^2\theta(d\phi+v_r dr
  +v_\theta d\theta)^ 2\right],
\end{equation}
where $\sigma, q, v_r$ and $v_\theta$ are regular functions of $r$ and
$\theta$. In these coordinates, the Killing vector is given by
\begin{equation}
\eta^i=(\partial_\phi)^i,
\end{equation}
and its square norm is
\begin{equation}
\eta=\eta^i\eta_i=e^{\sigma}r^2\sin^2\theta.
\end{equation}
The regularity conditions on the metric $h$ at the axis imply that
\begin{equation}
q|_\Gamma=0,
\end{equation}
where $\Gamma$ denotes the polar axis $r\sin\theta=0$.

The twist potential $\omega$ of the spacetime axial Killing field can be
computed in terms of the second fundamental form $K_{ij}$ as follows (see
\cite{Dain06c} for details). Define the vector $S^i$ by
\begin{equation}
 S_i=K_{ij}\eta^j-\eta^{-1}\eta_i K_{jk}\eta^j\eta^k,
\end{equation}
then, define $K_i$ by
\begin{equation}
 K_i=\epsilon_{ijk}S^j\eta^k,
\end{equation}
where $\epsilon_{ijk}$ is the volume element with respect to the flat
metric. In virtue of the constraint equations, the vector $K_i$ is
 the gradient of a scalar field which is the twist potential, namely
\begin{equation}\label{twist}
 K_i=\frac{1}{2}\nabla_i\omega.
\end{equation}

In this work, we will study the geometry of the 2-surfaces $r=constant$. For
these surfaces there exist two relevant quantities.  The first one is the
angular momentum, which for such surfaces is defined by
\begin{equation}\label{jota}
J=\frac{1}{8}\left(\omega(r,\theta=\pi)-\omega(r,\theta=0)\right).
\end{equation}
See \cite{Dain06c} for a detailed discussion about angular momentum in axial
symmetry.

The second quantity is the area of the surface. The induced 2-metric on such
surfaces is
\begin{equation}
\gamma=e^ \sigma r^2\left[e^ {2q}d\theta^2+\sin^2\theta (d\phi+v_\theta
  d\theta)^2\right], 
\end{equation}
and, remarkably, its determinant does not depend on $v_\theta$ or $v_r$,
\begin{equation}
\det(\gamma)=e^{2\sigma+2q}r^4\sin^2\theta.
\end{equation}
Then, the area of the surface $r=constant$ is given by
\begin{equation}\label{area}
A=\int_{S^2}\sqrt{\det(\gamma)}d\theta
d\phi=\int_{S^2}e^{\sigma+q}r^2dS=2\pi r^2\int_0^\pi e^{\sigma+q}\sin\theta
d\theta,
\end{equation}
where $dS=\sin\theta d\theta d\phi$ is the surface element on the unit 2-sphere
and in the last equality we have made use of axial symmetry.
 
Another useful geometrical quantity is the second fundamental form $\chi_{ij}$
of the surface given by
\begin{equation}\label{chi}
\chi_{ij}=-\gamma_i^k\nabla_k n_j
\end{equation}
where $n_i=e^{\sigma/2+q}(dr)_i$ is the unit normal to the surface. The mean
curvature of the surface is an important concept in what follows and reads
\begin{equation}\label{mean}
\chi=\chi^i_i=e^{-\sigma/2-q}\left(\partial_r(\sigma+q)+\frac{2}{r}\right),
\end{equation}

The mean curvature $\chi$ is related with the radial derivatives of the area
$A$ as follows. For the first derivative we have
\begin{equation}
  \label{eq:2}
 \partial_r A=\int_{S^2}e^{\frac{3}{2}\sigma+2q}r^2 \chi  \, dS, 
\end{equation}
and for the second derivative,
\begin{equation}
  \label{eq:3}
  \partial_r^2 A=\int_{S^2} \left[  e^{\frac{3}{2}\sigma+2q}r^2 \partial_r \chi
  +\chi \partial_r(r^2 e^{\frac{3}{2}\sigma+2q}) \right] \,dS.
\end{equation}
The case $\chi=0$ will be relevant for our purposes. For that case we have
\begin{equation}
  \label{eq:5}
  \partial_r^2 A=\int_{S^2} e^{\frac{3}{2}\sigma+2q}r^2 \partial_r \chi \, dS.
\end{equation}

The following theorem constitutes the main result of the present work.

\begin{theorem}
\label{teobh} 
Consider axisymmetric, vacuum and maximal initial data, with a non-negative
cosmological constant as described above. Assume there exists a surface
$\Sigma=\{r=constant\}$ where the following local conditions are
satisfied
\begin{align}
\chi  &=0,\label{cond1} \\ 
\partial_r\chi &\geq0, \label{cond2}\\
 \partial_rq &=0.\label{cond3}
\end{align}
Then we have 
\begin{equation}\label{desigualdad}
8\pi |J|\leq A
\end{equation}
where $A$ is the area and $J$ the angular momentum of $\Sigma$.
\end{theorem}

Let us discuss the hypothesis of this theorem.  The theorem is a pure local
result, in particular there are no conditions on the asymptotics of the initial
data. We have also introduced the cosmological constant, generalizing in this
way the validity of conjecture \ref{c:1}. As we mention in the
introduction, in \cite{Booth:2007wu} it has been presented a counter example to
inequality (\ref{eq:1}) with negative cosmological constant. Maximal data
with non-negative cosmological constant (as required in the hypothesis of
the theorem) have non-negative Ricci scalar. This is the crucial property that
allows us to prove (\ref{desigualdad}).

The first important restriction of the theorem is that only surfaces
$r=constant$ are allowed.  In the general case, the horizon mentioned in the
conjecture will not be such a surface. This particular choice of foliation
adapted to the cylindrical coordinates simplifies considerably the estimates. A
relevant open problem is how to extend these results to include general
surfaces.

By equations \eqref{eq:2} and \eqref{eq:5}, we deduce that conditions
\eqref{cond1} and \eqref{cond2} imply that the area of the surface $\Sigma$ is
a local minimum. That is, $\Sigma$ is a minimal surface \footnote{In the
  literature it is also common to call minimal a surface having $\chi=0$. In
  this article, we use the term extremal for such surface, and reserve the term
  minimal for surfaces which in addition are area minimizing.}. If these were
the only hypothesis in the theorem, then conjecture \ref{c:1} would be proved
for initial data having a global minimal surface $r=constant$: by definition,
the area of such minimal surface is less than or equal to the area of any
surface, the horizon in particular.  However, in order to prove the inequality
\eqref{desigualdad} we require the extra condition \eqref{cond3}. This is a
technical condition which we do not expect to be necessary. However, it is
important to emphasize that this is also a geometrical condition, since it can
be written in terms of $\chi$ and the component $\chi_{ij}\eta^i\eta^j$ of the
extrinsic curvature $\chi_{ij}$ of $\Sigma$, namely
\begin{equation}\label{chicero}
\chi=0\qquad \mbox{and} \qquad\chi_{ij}\eta^ i\eta^ j=0
\qquad\Rightarrow\qquad \partial_rq=0. 
\end{equation}  
This result can be seen in the following way
\begin{equation}
\chi_{ij}\eta^ i\eta^ j=-\gamma^k_i\nabla_k n_j\eta^i\eta^j=-\nabla_i
n_j\eta^i\eta^j=-\nabla_i(n_j\eta^j)\eta^i+n_j(\nabla_i\eta^j)\eta^i, 
\end{equation}
but the first term in the last expression is zero because $\eta^i$ and $n_i$
are orthogonal. Then
\begin{equation}
\chi_{ij}\eta^i\eta^j=n^j(\nabla_i\eta_j)\eta^i=-n^j(\nabla_j\eta_i)\eta^i
\end{equation}
since $\eta^i$ is a Killing vector field. Finally we have
\begin{equation}
  \chi_{ij}\eta^ i\eta^
  j=-\frac{1}{2}n^j\nabla_j\eta=-\frac{1}{2}n^r\partial_r\eta=-\frac{1}{2}n^rr^2
\sin^2\theta  e^{\sigma}\left(\partial_r\sigma+\frac{2}{r}\right),   
\end{equation}
where, in the second equality, we have made use of axial symmetry. Therefore,
$\chi_{ij}\eta^ i\eta^ j=0$ together with $\chi=0$, imply $\partial_rq=0$ (see
equation \eqref{mean}).

This alternative way of writing the hypotheses of theorem \ref{teobh} gives a
more geometrical description of the surface considered. In particular, totally
geodesic surfaces (i.e. surfaces such that $\chi_{ij}=0$) satisfy
\eqref{chicero} and \eqref{cond1}. Moreover, when $v_r\equiv v_\theta\equiv0$,
condition \eqref{chicero} implies that $\Sigma$ is totally geodesic.

There exists a particularly relevant class of initial data that satisfies all
conditions imposed in theorem \ref{teobh}.  Namely, initial data with an
isometry that leaves the surface $\Sigma$ invariant.  This isometry is also
called ``inversion through the throat'' in the literature.  Let us discuss this
family of examples in more detail.

In \cite{Gibbons72} it has been proven that a compact 2-surface that is
invariant under an isometry is a totally geodesic surface.  Note that the
isometry only imposes conditions on the first order derivatives of the initial
data functions evaluated at the invariant surface $\Sigma$.  Condition
(\ref{cond2}), which is a condition on the second derivatives of the metric, is
not automatically satisfied. In fact, a surface could be a local maximum and
still be invariant under the isometry. However, for a rich class of data this
surface is a global minimum. To analyze this point it is better to discuss
concrete examples.

A canonical example of this kind of isometric data is a slice $t=constant$
(in the standard Boyer-Lindquist coordinates) of the non-extreme Kerr black
hole. The geometry of these data (which is the same as the Schwarzschild black
hole) is the well known picture shown in figure
\ref{fig:isometry-initial-data}. The global minimal surface that connects the
two sheets satisfies all the hypothesis of the theorem. In this case the minimal
surface coincides with the apparent horizon. The Kerr black hole initial data
constitute a non-trivial example of the theorem. Although inequality
\eqref{desigualdad} can be of course computed explicitly for Kerr, the theorem
presents an alternative proof of it.

For more general black hole initial data, isometry conditions were introduced
in \cite{Misner63} and since then they have been extensively used to construct
initial data for black hole numerical simulations.  A well known example is
the conformally flat (i.e. $q=0$) family of black hole initial data introduced
by Bowen and York \cite{Bowen80}.  Another conformally flat case was analyzed
in \cite{Dain02c}.  Non-conformally flat examples have been studied in
\cite{Brandt95}. In all these examples black hole initial data with the geometry
shown in figure \ref{fig:isometry-initial-data} have been constructed.  These
data represent distorted black holes and can, in principle, be
far from equilibrium. However, it is important to emphasize that in order for
the isometry invariant surface $\Sigma$ to remain a global minimum of the area the
distortion should not produce extra minimal surfaces, as in the case shown in
figure \ref{fig:multiple-minimal-surfaces}. In that case the theorem still
applies for the surface $\Sigma$ but since it is does not give the global minimum of the area,
we can not prove the conjecture. 

Conditions ensuring that the isometry surface gives a minimum for the area have been studied for some examples in \cite{Bowen80}. There exists also a lot of
numerical evidence for all these examples showing that if the
distortion from Kerr is not too severe, then the isometry surface gives in fact a
global minimum for the area (for example, see \cite{Cook90a}).  It is
interesting to note that in many of these examples the horizon and the minimal
surface do not coincide \cite{Cook90a}. Also, this kind of
isometry data can represent binary black hole initial data (this is in fact
one of the main applications of these data, see \cite{Misner63},
\cite{Bowen80})). In the case of binary black holes, there exist two minimal
surfaces which are invariant under the isometry. Remarkable enough, these
surfaces also satisfy conditions (\ref{cond1}), (\ref{cond2}) and
(\ref{cond3}). But they are not $r=constant$ surfaces, and hence the theorem
does not apply to them. 

It is also important to mention that the isometry condition is preserved under
the evolution (in fact it has been explicitly used for numerical evolution, see
\cite{Brandt95}) and hence it can play a useful role in the analytical study
of the black hole stability problem for this kind of data.

We summarize the above discussion in the following corollary. 
\begin{corollary}
  Consider axisymmetric, vacuum and maximal initial data, with a
  non-negative cosmological constant as described above. Assume there exists an
  isometry which leaves fixed a surface $\Sigma=\{r=constant\}$. Assume also
  that the area of this surface is a global minimum. Then conjecture \ref{c:1}
  is proven for these data.
\end{corollary}
 
\begin{figure}[h]
 \centering
   \includegraphics[width=0.5\textwidth]{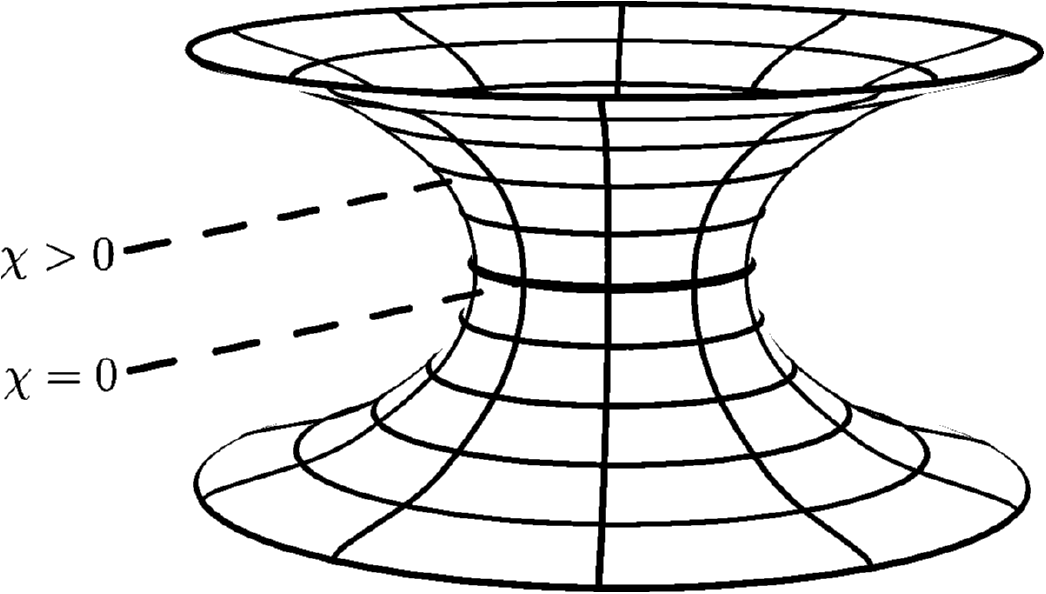}
 \caption{Initial data with an isometry. The data have a global minimal surface
 in the middle. The conjecture \ref{c:1} is proven for this kind of data. }
 \label{fig:isometry-initial-data}
\end{figure}

\begin{figure}[h]
 \centering
  \includegraphics[width=0.3\textwidth]{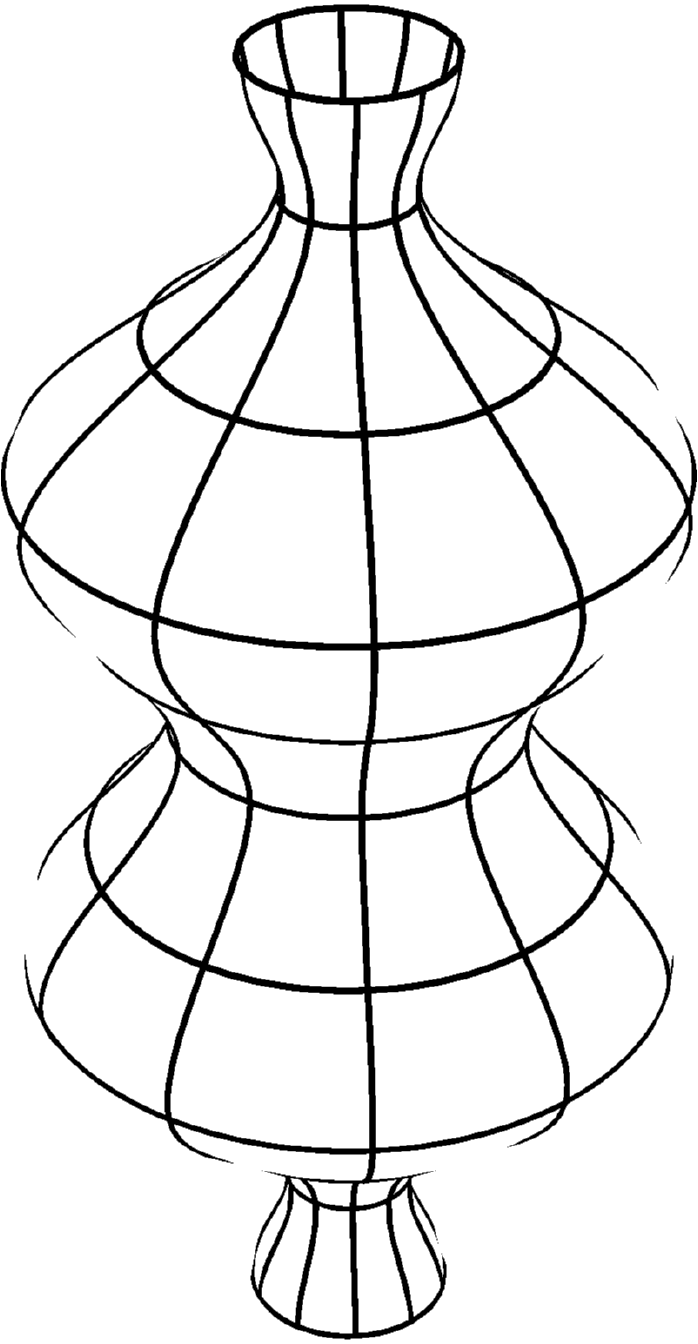}
 \caption{An example of data with an isometry. The theorem applies
   to the minimal surface in the middle, but the conjecture can not be proven in
 this case because there are other minimal surfaces with less area.}
 \label{fig:multiple-minimal-surfaces}
\end{figure}

There exists another class of initial data that is not strictly included in
the formulation of conjecture \ref{c:1}, but which plays an important role in the
proof as a limit case. And also, it could have further interesting applications.
These are the extreme throat initial data introduced in \cite{dain10d}. This
kind of data is a special case of axially symmetric data discussed above
where an additional symmetry is present. In order to introduce them, it is
convenient to make the change of coordinates $s=-\ln r$ to the metric
\eqref{metric}. We also define the function $\varsigma$ by
\begin{equation}\label{vsigma}
\varsigma= \sigma + 2\ln r. 
\end{equation}
Then, the line element \eqref{metric} is written as 
\begin{equation}
\label{metricCil}
  h=e^ {\varsigma}\left[e^ {2q}(ds^2+d\theta^ 2)+\sin^2\theta(d\phi+v_s ds
    +v_\theta d\theta)^ 2\right]. 
\end{equation}
Assume that $\varsigma$, $q$, $v_s$ and $v_\theta$ do not depend on $s$, then
$\partial_s$ is a Killing vector of $h$ (besides $\partial_\phi$). If also
$\mathcal L_{\partial_s} K_{ij}=0$ then we call $(S,h_{ij},K_{ij})$ an extreme
throat initial data set. The geometry of these data is 
cylindrical (see figure \ref{fig:cylinder-initial-data}), as $S$ is $S^2\times\mathbb{R}$ and the sections $s=constant$ are topological spheres which are
isometric to each other.  

\begin{figure}[h]
 \centering
  \includegraphics[width=0.2\textwidth]{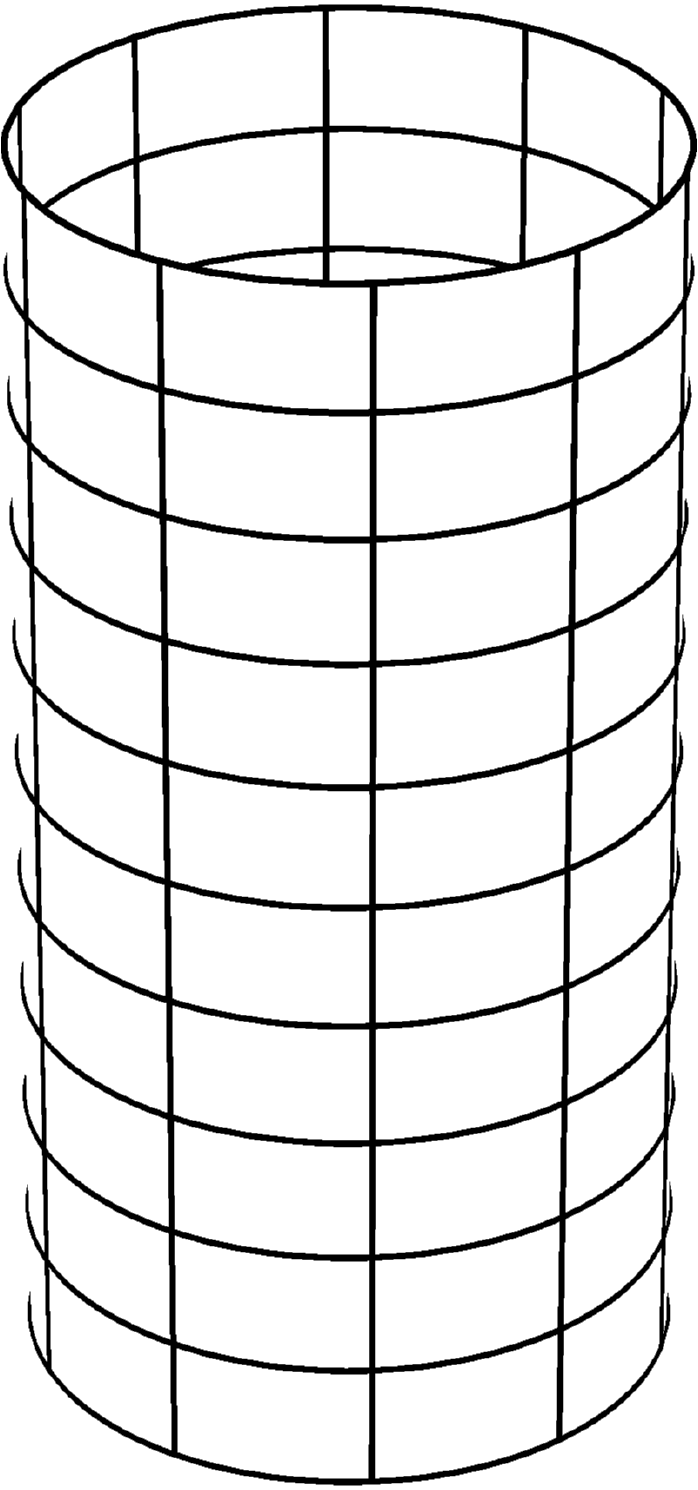}
 \caption{Cylindrical geometry of extreme throat initial data.}
 \label{fig:cylinder-initial-data}
\end{figure}

For these data the twist potential $\omega$ is defined in the same way as done
for general axially symmetric data, and now it is a function that depends only
on $\theta$. As the data are symmetric with respect to translations in the $s$
direction, the angular momentum and area associated with these data do not
depend on $s$. They are given by
\begin{equation}\label{jotaExt}
 J=\frac{1}{8}(\omega(\pi)-\omega(0)),
\end{equation}
\begin{equation}
\label{AExt}
 A=2\pi\int_0^\pi e^{\varsigma+q}\sin\theta\,d\theta,
\end{equation}
which are the corresponding expressions to \eqref{jota} and \eqref{area}.

A particularly relevant functional for this class of data is the mass
functional ${\cal M}$, defined as (see \cite{dain10d})
\begin{equation}
\label{eq:4}
 {\cal
   M}=\int_0^\pi\left(|\partial_\theta\varsigma|^2+
4\varsigma+\frac{|\partial_\theta\omega|^2}{\eta^2}\right)\sin\theta\,d\theta.   
\end{equation}
This functional plays a fundamental role in the proof of theorem
\ref{teobh} and has two important properties. The first one is
that it is possible to relate it with the area $A$. The second property is that
it is essentially equivalent to the energy of an harmonic map. Both properties
will constitute the core of the proof of theorem \ref{teobh}, they will be
explained in sections \ref{secarea} and \ref{MassMinimum} respectively.  

Note that by the non-dependence on $s$ (and hence on $r$) of the functions,
conditions \eqref{cond1}, \eqref{cond2} and \eqref{cond3} are satisfied for
extreme throat initial data. Hence we have the following corollary of theorem \ref{teobh}.

\begin{corollary}
\label{teocil}
  For an extreme throat initial data set the inequality 
\begin{equation}
 8\pi|J|\leq A.
\end{equation}
holds, where $A$ and $J$ are the area and angular momentum of the data given by
\eqref{jotaExt} and \eqref{AExt}.
\end{corollary}
This corollary significantly extends the results presented in \cite{dain10d} in
two directions. First, it applies to general extreme throat initial data. The
Killing vector $\partial_s$ is not required to be hypersurface
orthogonal. In fact, the result applies to a slightly more general class of data,
the only condition required is that the functions $\varsigma$, $q$ and
$\omega$ do not depend on $s$. But no condition is imposed on $v_s$ and
$v_\theta$. These functions could depend on $s$, in that case the data will not
admit the Killing field $\partial_s$ but the corollary will still hold.

The second extension with respect to \cite{dain10d} is that this is a global
result and not a local one. As we mentioned in the introduction, to prove this
result we will use a remarkable inequality relating the area and the mass
functional. This is explained in section \ref{secarea}.

The importance of extreme throat initial data resides on that they naturally
appear as the limit geometry of initial data with a cylindrical end. The
canonical example of these data is the extreme Kerr black hole initial data.
See figure \ref{fig:cylindrical-end-initial-data} for a representation of the
geometry of these data. At the cylindrical end, all the derivatives with
respect to $r$ of the relevant functions decay to zero, but not the functions
themselves. They have a well defined limit. These limit functions define extreme throat initial data (for the details of this construction see
\cite{dain10d}). 

Remarkably, the previous corollary applies to
the limit area of the cylindrical end of extreme black hole initial
data. This limit area is not directly related with conjecture \ref{c:1} because
it is not a horizon. But it still has interesting applications. 
Extreme Kerr is of course an example, where the equality holds. But there
are also other examples as the ones constructed in \cite{Dain:2008ck}
\cite{Dain:2008yu}   \cite{gabach09} \cite{Dain:2010uh}.
\begin{figure}[h]
 \centering
  \includegraphics[width=0.5\textwidth]{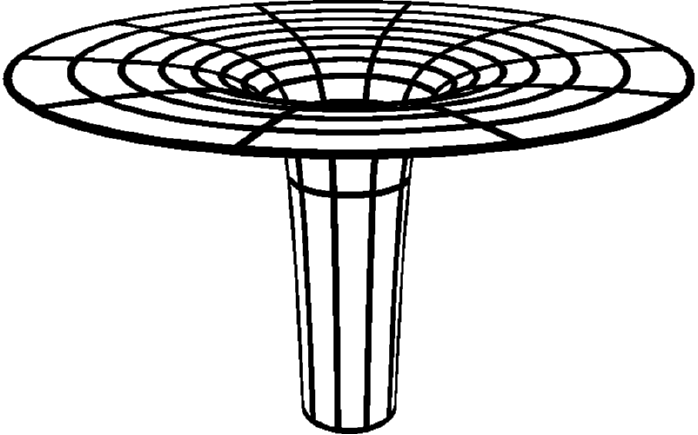}
 \caption{Geometry of extreme black hole initial data. The data have an asymptotically flat end (top) and a cylindrical end (bottom).}
 \label{fig:cylindrical-end-initial-data}
\end{figure}

Finally, we present the proof of theorem \ref{teobh}. 

\begin{proof}[Proof of theorem \ref{teobh}]

  The proof is divided into two main parts described in sections \ref{secarea}
  and \ref{MassMinimum}. First we obtain an inequality relating the area
  with the mass functional. This is given by lemma \ref{lema2}.  Here, we have
  made use of hypothesis (\ref{cond1}), (\ref{cond2}) and
  (\ref{cond3}). Then, using lemma \ref{lemma:2} we bound the mass functional
  by the angular momentum and we obtain the desired result.
\end{proof}

\section{Area and mass functional}\label{secarea}

The purpose of this section is to get a lower bound on the area $A$ of the surface $\Sigma$ mentioned in theorem \ref{teobh} in terms of the mass functional \eqref{eq:4}. All the calculations that follow are local, no asymptotic behavior being assumed.

The Hamiltonian equation \eqref{hamiltonianConst}, together with the maximality condition, $K=0$, give
\begin{equation}\label{maxConst}
R=K_{ij}K^{ij}+2\Lambda.
\end{equation}
In \cite{Dain06c} it has been proven that
\begin{equation}\label{cotak}
K_{ij}K^{ij}\geq \frac{1}{2}\frac{|\partial\omega|^2}{\eta^2}e^{-\sigma-2q},
\end{equation}
where $\omega$ is the twist potential introduced in \eqref{twist}. By inserting \eqref{cotak} in \eqref{maxConst} and using the condition $\Lambda\geq0$, one gets the bound
\begin{equation}\label{ineqRicci}
R\geq\frac{1}{2}\frac{|\partial\omega|^2}{\eta^2}e^{-\sigma-2q}.
\end{equation}
The Ricci scalar $R$ in terms of the metric functions $\sigma$, $q$, $v_r$ and $v_\theta$ is
\begin{equation}\label{ricci}
R=-2e^ {-\sigma-2q}\Bigg(\Delta\sigma+\Delta_2q+\frac{1}{4}|\partial\sigma|^2+\frac{1}{4}\sin^2\theta e^{-2q}(v_{r,\theta}-v_{\theta,r})^2\Bigg),
\end{equation}
where $\Delta$ and $\Delta_2$ are the flat Laplace operators in three and two dimensions respectively, and $\partial$ denotes partial derivatives. From \eqref{ineqRicci} and \eqref{ricci} we finally obtain
\begin{equation}\label{funda}
|\partial\sigma|^2+\frac{|\partial\omega|^2}{\eta^2}\leq-4\left(\Delta\sigma+\Delta_2q\right).
\end{equation}

We write this inequality explicitly in spherical coordinates and arrange terms in the following useful way
\begin{equation}\label{ecfunesf}
(\partial_\theta\sigma)^2+\frac{(\partial_\theta\omega)^2}{\eta^2}\leq-4(\partial_\theta^2q+\Delta_0\sigma)-f-g
\end{equation}
where $\Delta_0$ is the Laplace operator on $S^2$ acting on axially symmetric functions
\begin{equation}
\Delta_0\sigma=\frac{1}{\sin\theta}\partial_\theta(\sin\theta\partial_\theta \sigma),
\end{equation}
and
\begin{align}
\label{deff}& f=4r^2\partial^2_r(\sigma+q),\\
\label{defg}& g=4r\partial_r(2\sigma+q)+r^2(\partial_r\sigma)^2+r^ 2\frac{(\partial_r\omega)^2}{\eta^2}.
\end{align}
Now we integrate \eqref{ecfunesf} in the unit 2-sphere. Using
\begin{equation}
\int_{S^2}\Delta_0\sigma dS=0
\end{equation} 
and (see \cite{dain10d})
\begin{equation}
\int_{S^2}\partial_\theta^2q dS=-\int_{S^2} q  dS,
\end{equation}
we obtain
\begin{equation}\label{cota1}
\int_{S^2}(\partial_\theta\sigma)^2+\frac{(\partial_\theta\omega)^2}{\eta^2}dS\leq4\int_{S^2}qdS-F-G,
\end{equation}
where we have defined 
\begin{equation}\label{defFG}
 F=\int_{S^2}fdS,\qquad G=\int_{S^2}gdS.
\end{equation}
In order to make contact with the mass functional $\mathcal M$, we write the left hand side of inequality \eqref{cota1} in terms of $\varsigma$, defined in \eqref{vsigma},
\begin{equation}\label{cota2}
\int_{S^2}(\partial_\theta\varsigma)^2+\frac{(\partial_\theta\omega)^2}{\eta^2}dS\leq4\int_{S^2}qdS-F-G.
\end{equation}
Adding up a term of the form $4\int_{S^2}\varsigma dS$ to both sides of \eqref{cota2} we obtain an upper bound for $\mathcal M$,
\begin{equation}\label{cota3}
2\pi \mathcal M\leq4\int_{S^2}(\varsigma+q)dS-F-G.
\end{equation}
This is more conveniently written as
\begin{equation}\label{cota4}
\frac{\mathcal M}{8}+\frac{F+G}{16\pi}\leq\frac{1}{2}\int_0^\pi(\varsigma+q)\sin\theta d\theta,
\end{equation}
where we have integrated the right hand side over $\phi$. For our purposes later, it is useful to exponentiate the above inequality 
\begin{equation}\label{cotaexp}
e^{\frac{\mathcal M}{8}}e^{\frac{F+G}{16\pi}}\leq e^{\frac{1}{2}\int_0^\pi(\varsigma+q)\sin\theta d\theta}.
\end{equation}

We want to remark that in going from \eqref{cota1} to \eqref{cota2} there appears to be an inconsistency concerning units. Nevertheless, since in the end we take the exponential of that inequality, the final result, \eqref{cotaexp}, has the right dimension of area.

We want to use this inequality to bound the area of the surfaces $r=constant$. We write the area for these surfaces in the form (see equation \eqref{area})
\begin{equation}
\frac{A}{4\pi }=\frac{1}{2}\int_0^\pi e^{\varsigma+q}\sin\theta d\theta,
\end{equation}
then, using Jensen's inequality for the exponential function we have
\begin{equation}
\frac{1}{2}\int_0^\pi e^{\varsigma+q}\sin\theta d\theta\geq e^{\frac{1}{2}\int_0^\pi\varsigma+q \sin\theta d\theta}
\end{equation}
which gives
\begin{equation}\label{area1}
\frac{A}{4\pi}\geq e^{\frac{1}{2}\int_0^\pi\varsigma+q \sin\theta d\theta}.
\end{equation}
Finally, we put together inequalities \eqref{cotaexp} and \eqref{area1} to obtain
\begin{equation}\label{area2}
A\geq4\pi e^{\frac{\mathcal M}{8}}e^{\frac{F+G}{16\pi}},
\end{equation}
and this gives a bound for the area of the surface $r=constant$ through the functional $\mathcal M$ and radial derivatives of $\sigma$ and $q$.
 
This result proves the following lemma.
\begin{lemma}\label{lema1}
The area $A$ of a surface $r=constant$ satisfies inequality \eqref{area2}, where $F$ and $G$ are defined in \eqref{defFG} and $\mathcal M$ is evaluated at $r$.
\end{lemma} 
 
We remark that in order to get inequality \eqref{area2} we have only made use of the Hamiltonian  constraint, the maximality condition, the positivity of the cosmological constant $\Lambda$ and the bound \eqref{cotak} for the square of the extrinsic curvature. Moreover, it holds for any surface $r=constant$. 

Inequality \eqref{area2} is the crucial ingredient in proving the following lemma.

\begin{lemma}\label{lema2}
If conditions \eqref{cond1}-\eqref{cond3} of theorem \ref{teobh} hold for the surface $\Sigma$ of constant $r$, then its area $A$ satisfies
\begin{equation}
A\geq4\pi e^{\frac{(\mathcal M-8)}{8}}
\end{equation}
where $\mathcal M$ is the mass functional \eqref{eq:4} evaluated at the surface $\Sigma$.
\end{lemma}
\begin{proof}
Take inequality \eqref{area2} and evaluate it at the particular surface $\Sigma$ mentioned in theorem \ref{teobh}. Note that \eqref{cond1} and \eqref{cond3} restrict the first order radial derivatives of $\sigma$ and $q$, while \eqref{cond2} gives a condition on the second order radial derivative of $\sigma+q$. More precisely, in virtue of \eqref{cond1} and \eqref{cond2} we have
\begin{equation}
F|_\Sigma=\int_\Sigma 4r^ 2\partial_r^ 2(\sigma+q)dS\geq32\pi
\end{equation}
and conditions \eqref{cond1} and \eqref{cond3} give
\begin{equation}
G|_\Sigma=\int_{\Sigma}\left[4r\partial_r(2\sigma+q)+r^2(\partial_r\sigma)^2+r^ 2\frac{(\partial_r\omega)^2}{\eta^2}\right]dS\geq-48\pi.
\end{equation}
Therefore, inequality \eqref{area2} evaluated at $\Sigma$ gives
\begin{equation}\label{area4}
A\geq4\pi e^{\frac{(\mathcal M-8)}{8}},
\end{equation}
where $\mathcal M$ must also be evaluated at $\Sigma$.
\end{proof}

\section{Extreme Kerr is a global minimum of ${\cal M}$}\label{MassMinimum}

The purpose of this section is to show that the extreme Kerr throat initial data is a global minimum of the mass functional \eqref{eq:4}, which completes the proof of theorem \ref{teobh}.

The extreme Kerr throat initial data depend only on one parameter, the angular momentum $J$, and is given by (using a subscript `$0$' to indicate that we refer to this particular case)
\begin{equation}\label{datos}
 \varsigma_0=\ln(4|J|)-\ln(1+\cos^2\theta),\hspace{.5cm}\omega_0=-\frac{8J\cos\theta}{1+\cos^2\theta},\hspace{.5cm}q_0=\ln\frac{1+\cos^2\theta}{2}.
\end{equation}
See \cite{dain10d} for details. Evaluating \eqref{AExt} for these data we have
\begin{equation}
 A_0=8\pi|J|,
\end{equation}
and evaluating \eqref{eq:4}
\begin{equation}
 \mf_0=8(\ln(2|J|)+1).
\end{equation}

We need further properties of the functional $\mf$. On the unit sphere, using $D$ to denote the covariant derivative with respect to the standard metric on $S^2$, this functional takes the form
\begin{equation}\label{defM}
 {\cal M}=\frac{1}{2\pi}\int_{S^2}\left(|D\varsigma|^2+4\varsigma+\frac{|D\omega|^2}{\eta^2}\right)dS,
\end{equation}
where as before $\eta=e^\varsigma\sin^2\theta$. The Euler-Lagrange equations for $\cal M$ are
\begin{equation}\label{euler}
 D_AD^A\varsigma-2=\frac{D_A\omega D^A\omega}{\eta^2},\hspace{2cm}D_A\left(\frac{D^A\omega}{\eta^2}\right)=0.
\end{equation}
It is important to know that extreme Kerr throat initial data \eqref{datos} satisfy the Euler Lagrange equations \eqref{euler}.

A crucial property of the functional $\cal{M}$ is that it is closely related to the energy associated with a particular harmonic map. To see this, let us restrict the domain of integration in \eqref{defM}, defining
\begin{equation}
 {\cal M}_\Omega=\frac{1}{2\pi}\int_{\Omega}\left(|D\varsigma|^2+4\varsigma+\frac{|D\omega|^2}{\eta^2}\right)dS,
\end{equation}
where $\Omega\subset S^2$, such that $\Omega$ does not include the poles, and consider the functional
\begin{equation}
 \tilde{\cal M}_\Omega=\frac{1}{2\pi}\int_\Omega\frac{|\partial\eta|^2+|\partial\omega|^2}{\eta^2}dS.
\end{equation}
The relation between ${\cal M}_\Omega$ and $\tilde{\cal M}_\Omega$ is given by
\begin{equation}
 \label{relMMt}
 \tilde{\cal M}_\Omega={\cal M}_\Omega+4\int_\Omega\log\sin\theta\,dS+\oint_{\partial\Omega}(4\varsigma+\log\sin\theta)\frac{\partial\log\sin\theta}{\partial n}ds,
 \end{equation}
where $n$ denotes the exterior normal to $\Omega$ and $ds$ is the surface element on the boundary $\partial\Omega$. The second term on the r.h.s. is a non divergent numerical constant, 
but the boundary term diverges at the poles.

The functional $\tilde{\cal M}$ defines an energy for maps $(\eta,\omega):S^2\rightarrow \mathbb{H}^2$, where $\mathbb{H}^2$ denotes the hyperbolic plane, that is, $\mathbb{H}^2=\{(\eta,\omega):\eta>0\}$ equipped with the negative constant curvature metric
\begin{equation}
 ds^2=\frac{d\eta^2+d\omega^2}{\eta^2}.
\end{equation}
Solutions to the Euler-Lagrange equations for the energy $\tilde{\cal M}$ are called harmonic maps from $S^2\rightarrow \mathbb{H}^2$. Since ${\cal M}$ and $\tilde{\cal M}$ differ only by a constant and boundary terms, they have the same Euler-Lagrange equations. Relation \eqref{relMMt} will be central in the proof that the global minimum of $\cal{M}$ is attained by extreme Kerr throat initial data. This result is presented in the following lemma.

\begin{lemma}
 \label{lemma:2}
Let $\varsigma$ and $\omega$ regular functions on the sphere such that the functional $\mathcal M$ is finite. Assume also that $\partial_\theta \omega=0$ for $\theta=0,\pi$.  Then
\begin{equation}
 \mf\geq8(\ln(2|J|)+1).
\end{equation}
where $J$ is defined in terms of $\omega$ by \eqref{jotaExt}.
\end{lemma}

It is important to remark that the condition $\partial_\theta \omega=0$ for $\theta=0,\pi$ is automatically satisfied for smooth initial data, as a consequence of the regularity at the axis.
\begin{proof}
The proof follows similar arguments as those used in \cite{Costa:2009hn} and \cite{Chrusciel:2007ak}.

The core of the proof is the use of a theorem due to Hildebrandt, Kaul and
Widman \cite{Hildebrandt77} for harmonic maps. In that work it is shown that if
the domain for the map is compact, connected, with nonvoid boundary and the
target manifold has negative sectional curvature, then minimizers of the
harmonic energy with Dirichlet boundary conditions exist, are smooth, and
satisfy the associated Euler-Lagrange equations. That is, harmonic maps are
minimizers of the harmonic energy for given Dirichlet boundary
conditions. Also, solutions of the Dirichlet boundary value problem are unique
when the target manifold has negative sectional curvature. Therefore, we want
to use the relation between ${\cal M}$ and the harmonic energy $\tilde{\cal M}$
in order to prove that minimizers of $\tilde{\cal M}$ are also minimizers of
${\cal M}$. There are two main difficulties in doing this. First, the harmonic
energy $\tilde{\cal M}$ is not defined for the functions that we are
considering if the domain of integration includes the poles. Second, we are not
dealing with a Dirichlet problem. To overcome this difficulties the sphere is
split in three regions according to figure \ref{regions}. The extent of the
different regions depend on a chosen positive constant $\epsilon$, in such a
way that when $\epsilon$ goes to zero regions $\Omega_I$ and $\Omega_{II}$
shrink towards the poles, while region $\Omega_{III}$ extends towards covering
the sphere. Then a partition function is used to interpolate between extreme
Kerr throat initial data in region $\Omega_I$ and general extreme throat
initial data in region $\Omega_{III}$, constructing auxiliary interpolating
data. This solves the two difficulties in the sense that now the Dirichlet
problem on region $\Omega_{IV}=\Omega_{II}\cup\Omega_{III}$ can be considered,
and the harmonic energy is well defined for this domain of integration. This
allows us to show that the mass functional for Kerr data is less than or equal
to the mass functional for the auxiliary interpolating data in the hole
sphere. The final step is to show that as $\epsilon$ goes to zero the mass
functional for the auxiliary data converges to the mass functional for the
original general data. There is a subtlety in this step. Namely, the partition
function needs to have been chosen suitably in order for the convergence to be
possible.

\begin{figure}[h]
 \centering
 \includegraphics[width=4cm]{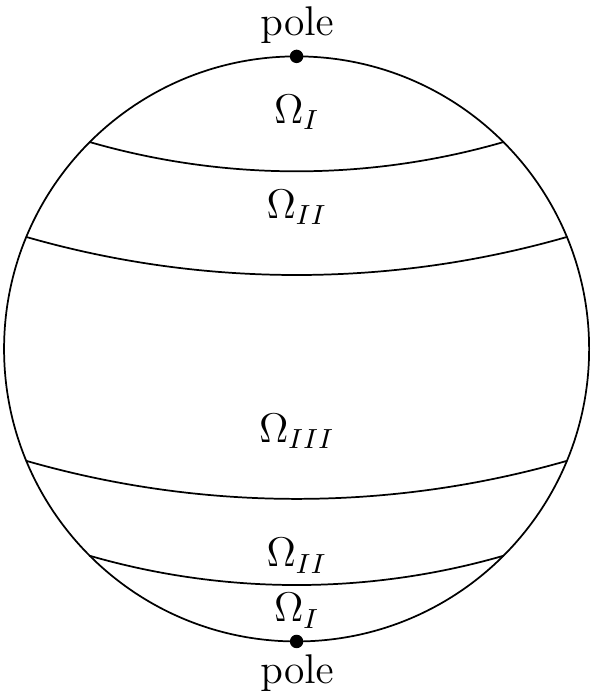}
 \caption{The different regions in which the sphere is split.}
 \label{regions}
\end{figure}

The partition function that will be used and the extension of the regions in
which the sphere is split are closely related. Therefore we start presenting
the partition function, taken from \cite{Li92}, Lemma 3.1, and then define the
different regions on the sphere. So, let $\chi:\mathbb{R}\rightarrow\mathbb{R}$
be a cut off function such that $\chi\in C^\infty(\mathbb{R})$,
$0\leq\chi\leq1$, $\chi(t)=1$ for $t\leq1$, $\chi(t)=0$ for $2\leq t$ and
$\left|\frac{d\chi}{dt}\right|\leq1$. For given $\epsilon>0$ define
\begin{equation}
 t_\epsilon(\rho)=\frac{\log(-\log\rho)}{\log(-\log\epsilon)},
\end{equation}
and
\begin{equation}
 \chi_\epsilon(\rho)=\chi(t_\epsilon(\rho)).
\end{equation}
Then $\chi_\epsilon$ defines a smooth function for $0<\epsilon<1$ and
$0\leq\rho<\infty$ (the function is trivially extended to be $1$ when
$\rho\geq1$). Also, $\chi_\epsilon(\rho)=0$ for $\rho\leq
e^{-(\log\epsilon)^2}$ and $\chi_\epsilon(\rho)=1$ for $\rho\geq\epsilon$, so
it has all the properties of a partition function. Of particular importance for
us is that $\chi_\epsilon$ has the following property,
\begin{equation}\label{intChi}
 \lim_{\epsilon\rightarrow0}\int_0^\infty|\partial_\rho\chi_\epsilon|^2\rho d\rho=0.
\end{equation}

Accordingly to the definition of $\chi_\epsilon$ we define the following regions on the sphere
\begin{eqnarray}
&& \Omega_I = \{\sin\theta\leq e^{-(\log\epsilon)^2}\},\\
&& \Omega_{II} = \{e^{-(\log\epsilon)^2}\leq\sin\theta\leq\epsilon\},\\
&& \Omega_{III} = \{\epsilon\leq\sin\theta\},\\
&& \Omega_{IV} = \Omega_{II}\cup\Omega_{III}.
\end{eqnarray}

Now we can define the interpolating functions. For this, let $u$ represent any of $(\varsigma,\omega)$, the general data, and let $u_0$ represent any of $(\varsigma_0,\omega_0)$, corresponding to the extreme Kerr throat data with the same angular momentum $J$ of $\omega$. We define $u_\epsilon$ to be
\begin{equation}
\label{uepsilon}
u_\epsilon=\chi_\epsilon(\sin\theta)\,u+(1-\chi_\epsilon(\sin\theta))\,u_0.
\end{equation}
This gives $u_\epsilon|_{\Omega_I}=u_0|_{\Omega_I}$ and $u_{\epsilon}|_{\Omega_{III}}=u|_{\Omega_{III}}$ as desired. We also define the mass functional for these functions
\begin{equation}
\label{Mepsilon}
{\cal
  M}^\epsilon=\frac{1}{2\pi}\int_{S^2}\left(|D\varsigma_\epsilon|^2+4\varsigma_\epsilon+\frac{|D\omega_\epsilon|^2}{e^{2\varsigma_\epsilon}\sin^4\theta}\right)dS,   
\end{equation}
and correspondingly ${\cal M}^\epsilon_\Omega$ and $\tilde{\cal
  M}^\epsilon_\Omega$ when the domain of integration is restricted to some region $\Omega$ or when we are considering the harmonic energy for the given map. We also denote by a superscript `$0$' these quantities
calculated for $u_0$.

We have all the ingredients needed to make use of the result of \cite{Hildebrandt77}. For this, let us consider now a fixed value of $\epsilon$, and the functions
$(\varsigma,\omega)$ on the set $\Omega_{IV}$. By \cite{Hildebrandt77} we know that there exists one and only one function that minimizes $\tilde{\cal M}$ on $\Omega_{IV}$ for
given boundary data, and that this function satisfies the Euler-Lagrange equations of $\tilde{\cal M}$ on $\Omega_{IV}$. By construction of $u_\epsilon$ we have that $u_\epsilon$ and $u_0$ have the same boundary values on $\Omega_{IV}$,
\begin{equation}
u_\epsilon|_{\partial\Omega_{IV}}=u_0|_{\partial\Omega_{IV}}.
\end{equation}
As we already know that $u_0$ is a solution of the Euler-Lagrange equations
of $\tilde{\cal M}$ there, then $u_0$ is the only minimizer of $\tilde{\cal M}$ on $\Omega_{IV}$ with these boundary conditions. With respect to $u_\epsilon$ this means that
\begin{equation}\label{desMepsilonMtalpha}
\tilde{\cal M}^\epsilon_{\Omega_{IV}}\geq\tilde{\cal M}^0_{\Omega_{IV}}.
\end{equation}
Both ${\cal M}$ and $\tilde{\cal M}$ are well defined on $\Omega_{IV}$, and by \eqref{relMMt} their difference is just a constant. This allows us to use \eqref{desMepsilonMtalpha} to get that
\begin{equation}\label{desMeps}
{\cal M}^\epsilon_{\Omega_{IV}}\geq{\cal M}^0_{\Omega_{IV}}.
\end{equation}
As we have already noted, $u_\epsilon|_{\Omega_{I}}=u_0|_{\Omega_{I}}$, and therefore ${\cal M}^\epsilon_{\Omega_I}={\cal M}^0_{\Omega_I}$. This together with \eqref{desMeps} and the fact that $S^2=\Omega_I\cup\Omega_{IV}$ gives
\begin{equation}\label{ineqMe}
{\cal M}^\epsilon\geq{\cal M}^0.
\end{equation}

Only the last step of the proof is lacking, that is, to show that
\begin{equation}\label{limitMe}
\lim_{\epsilon\rightarrow 0}{\cal M}^\epsilon={\cal M}.
\end{equation}
To do this we start by splitting the integral in \eqref{Mepsilon} according to the different domains of integration $\Omega_I$, $\Omega_{II}$ and $\Omega_{III}$. From the definition of $u_\epsilon$ \eqref{uepsilon} we have
\begin{equation}\label{Duepsilon}
 Du_\epsilon=D(\chi_\epsilon(\sin\theta))(u-u_0)+\chi_\epsilon(\sin\theta) Du+(1-\chi_\epsilon(\sin\theta))Du_0,
\end{equation}
and in particular $Du_\epsilon|_{\Omega_I}=Du_0|_{\Omega_I}$ and $Du_\epsilon|_{\Omega_{III}}=Du|_{\Omega_{III}}$. Then
\begin{eqnarray}
  {\cal M}^\epsilon  & = & \frac{1}{2\pi}\int_{\Omega_{I}}\left(|D\varsigma_0|^2+4\varsigma_0+\frac{|D\omega_0|^2}{e^{2\varsigma_0}\sin^4\theta}\right)dS\\
\label{intOm2} &&+\frac{1}{2\pi}\int_{\Omega_{II}}\left(|D\varsigma_\epsilon|^2+4\varsigma_\epsilon+\frac{|D\omega_\epsilon|^2}{e^{2\varsigma_\epsilon}\sin^4\theta}\right)dS\\
\label{intOm3} &&+\frac{1}{2\pi}\int_{\Omega_{III}}\left(|D\varsigma|^2+4\varsigma+\frac{|D\omega|^2}{e^{2\varsigma}\sin^4\theta}\right)dS.
\end{eqnarray}
The first and third integrals are not hard to deal with. As ${\cal M}_0$ is finite, and as $\Omega_I$ shrinks to a point as $\epsilon$ goes to zero, then
\begin{equation}
 \lim_{\epsilon\rightarrow 0}\left[\frac{1}{2\pi}\int_{\Omega_{I}}\left(|D\varsigma_0|^2+4\varsigma_0+\frac{|D\omega_0|^2}{e^{2\varsigma_0}\sin^4\theta}\right)dS\right]=0
\end{equation}
by the Lebesgue-dominated convergence theorem\footnote{The Lebesgue-dominated convergence theorem is usually stated in terms of a pointwise convergent sequence of functions, all of which are dominated by an integrable function. In our case, the functions are not changing as we take the limit, but the domain of integration itself is changing. For using the Lebesgue-dominated convergence theorem we can construct a sequence of functions, which are defined by taking the value of the original function in the domain that we are integrating, and zero outside this domain. So now we can keep the domain of integration fixed and apply the theorem in its most common form.}. Also, as ${\cal M}$ is finite, and as $\Omega_{III}$ extends to cover $S^2$ as $\epsilon$ goes to zero, then by the dominated convergence theorem we have that
\begin{eqnarray}
&& \lim_{\epsilon\rightarrow 0}\left[\frac{1}{2\pi}\int_{\Omega_{III}}\left(|D\varsigma|^2+4\varsigma+\frac{|D\omega|^2}{e^{2\varsigma}\sin^4\theta}\right)dS\right]=\\
&& \hspace{2cm}=\frac{1}{2\pi}\int_{S^2}\left(|D\varsigma|^2+4\varsigma+\frac{|D\omega|^2}{e^{2\varsigma}\sin^4\theta}\right)dS={\cal M}.
\end{eqnarray}

As can be expected, the integral over $\Omega_{II}$ \eqref{intOm2} is more tricky and we need to consider its different parts separately. Using the definition of $u_\epsilon$ \eqref{uepsilon} we have that $|u_\epsilon|\leq|u|+|u_0|$. Therefore
\begin{equation}
 \left|\int_{\Omega_{II}}\varsigma_\epsilon dS \right|\leq\int_{\Omega_{II}}|\varsigma_\epsilon|dS\leq\int_{\Omega_{II}}(|\varsigma|+|\varsigma_0|)dS,
\end{equation}
and by the Lebesgue-dominated convergence theorem
\begin{equation}
 \lim_{\epsilon\rightarrow 0}\int_{\Omega_{II}}(|\varsigma|+|\varsigma_0|)dS=0,
\end{equation}
so we get that
\begin{equation}
 \lim_{\epsilon\rightarrow 0}\left|\int_{\Omega_{II}}\varsigma_\epsilon dS\right|=0.
\end{equation}

From \eqref{Duepsilon} we have
\begin{equation}
 |Du_\epsilon|\leq|D(\chi_\epsilon(\sin\theta))|\,|u-u_0|+|Du|+|Du_0|,
\end{equation}
then
\begin{equation}\label{boundDu2}
 |Du_\epsilon|^2\leq3\left(|D(\chi_\epsilon(\sin\theta))|^2\,|u-u_0|^2+|Du|^2+|Du_0|^2\right).
\end{equation}
For the first term in \eqref{intOm2} this gives
\begin{eqnarray}
\label{intPart1} \int_{\Omega_{II}}|D\varsigma_\epsilon|^2dS & \leq & 3\bigg(\int_{\Omega_{II}}|D(\chi_\epsilon(\sin\theta))|^2\,|\varsigma-\varsigma_0|^2dS\\
&& +\int_{\Omega_{II}}|D\varsigma|^2dS+\int_{\Omega_{II}}|D\varsigma_0|^2dS\bigg).
\end{eqnarray}
The last two terms in the integral above converge to zero when $\epsilon$ goes to zero by the Lebesgue-dominated convergence theorem as $\cal{M}$ and ${\cal{M}}_0$ are finite. For the first term we have
\begin{equation}\label{est1}
 \int_{\Omega_{II}}|D(\chi_\epsilon(\sin\theta))|^2\,|\varsigma-\varsigma_0|^2dS\leq\sup_{S^2}\left((|\varsigma|+|\varsigma_0|)^2\right)\int_{\Omega_{II}}|D(\chi_\epsilon(\sin\theta))|^2dS.
\end{equation}
Making the change of variable $\rho=\sin\theta$ and integrating over the $\phi$ variable we have that
\begin{equation}\label{intPar}
 \int_{\Omega_{II}}|D(\chi_\epsilon(\sin\theta))|^2dS\leq2\pi\int_{e^{-(\log\epsilon)^2}}^\epsilon|\partial_\rho(\chi_\epsilon(\rho))|^2\rho d\rho.
\end{equation}
Here is where the important property \eqref{intChi} is used, together with \eqref{intPar}, \eqref{est1} and \eqref{intPart1} gives
\begin{equation}
 \lim_{\epsilon\rightarrow 0}\int_{\Omega_{II}}|D\varsigma_\epsilon|^2dS=0.
\end{equation}

The last integral that needs consideration is
\begin{equation}
\int_{\Omega_{II}}\frac{|D\omega_\epsilon|^2}{e^{2\varsigma_\epsilon}\sin^4\theta}dS.
\end{equation}
As $\varsigma$ and $\varsigma_0$ are bounded, then also is $\varsigma_\epsilon$  and therefore there exists a constant $C$ such that $e^{-2\varsigma_\epsilon}\leq C$, and then using \eqref{boundDu2}
\begin{eqnarray}
 \int_{\Omega_{II}}\frac{|D\omega_\epsilon|^2}{e^{2\varsigma_\epsilon}\sin^4\theta}dS & \leq &  3C\bigg(\int_{\Omega_{II}}\frac{|D(\chi_\epsilon(\sin\theta))|^2|\omega-\omega_0|^2}{\sin^4\theta}dS\\
&& +\int_{\Omega_{II}}\frac{|D\omega|^2}{\sin^4\theta}dS+\int_{\Omega_{II}}\frac{|D\omega_0|^2}{\sin^4\theta}dS\bigg).
\end{eqnarray}
As before, the last two terms in the integral converge to zero when $\epsilon$
goes to zero by the Lebesgue-dominated convergence theorem as $\cal{M}$ and
${\cal{M}}_0$ are finite. For the first term we use the hypothesis
$\partial_\theta \omega|_{\theta=0,\pi}=0$ and the fact that by construction
$\omega|_{\theta=0,\pi}=\omega_0|_{\theta=0,\pi}$. Making a Taylor expansion of
$\omega$, what we have just said translates into
$\omega=\omega_0+O(\sin^2\theta)$ near the axis. This means that
$\frac{|\omega-\omega_0|^2}{\sin^4\theta}$ is bounded, and together with
\eqref{intPar} and \eqref{intChi} we conclude that
\begin{equation}
 \lim_{\epsilon\rightarrow0}\int_{\Omega_{II}}\frac{|D\omega_\epsilon|^2}{e^{2\varsigma_\epsilon}\sin^4\theta}dS=0.
\end{equation}
This completes the proof of \eqref{limitMe} and with \eqref{ineqMe} the proof of the lemma.
\end{proof}

\label{conclusiones}
\section*{Acknowledgements}
Most of this work took place during the visit of A. A. to FaMAF, UNC, in 2010. He
thanks  for the hospitality and support of  this institution.

S. D. is supported by CONICET (Argentina). M. E. G. C. is supported by a
fellowship of CONICET (Argentina). This work was supported in part by grant PIP
6354/05 of CONICET (Argentina), grant Secyt-UNC (Argentina) and the
Partner Group grant of the Max Planck Institute for Gravitational Physics,
Albert-Einstein-Institute (Germany).

\appendix


\begin{thebibliography}{10}

\bibitem{Booth:2007wu}
Ivan Booth and Stephen Fairhurst.
\newblock {Extremality conditions for isolated and dynamical horizons}.
\newblock {\em Phys. Rev.}, D77:084005, 2008.

\bibitem{Bowen80}
Jeffrey~M. Bowen and James~W. York, Jr.
\newblock Time-asymmetric initial data for black holes and black-hole
  collisions.
\newblock {\em Phys. Rev. D}, 21(8):2047--2055, 1980.

\bibitem{Brandt95}
S.~Brandt and E.~Seidel.
\newblock Evolution of distorted rotating black holes i: Methods and tests.
\newblock {\em Phys. Rev. D}, 52(2):856--869, 1995.

\bibitem{Chrusciel:2007dd}
Piotr~T. Chrusciel.
\newblock {Mass and angular-momentum inequalities for axi-symmetric initial
  data sets I. Positivity of mass}.
\newblock {\em Annals Phys.}, 323:2566--2590, 2008.

\bibitem{Chrusciel:2007ak}
Piotr~T. Chru{\'s}ciel, Yanyan Li, and Gilbert Weinstein.
\newblock Mass and angular-momentum inequalities for axi-symmetric initial data
  sets. {II}. {A}ngular-momentum.
\newblock {\em Ann. Phys.}, 323(10):2591--2613, 2008.

\bibitem{gabach09}
Mar\'ia E~Gabach Cl\'ement.
\newblock Conformally flat black hole initial data with one cylindrical end.
\newblock {\em Classical and Quantum Gravity}, 27(12):125010, 2010.

\bibitem{Cook90a}
G.~Cook and J.~W. York.
\newblock Apparent horizons for boosted or spinning black holes.
\newblock {\em Phys. Rev. D}, 41(4):1077--1085, 1990.

\bibitem{Costa:2009hn}
Joao~Lopes Costa.
\newblock Proof of a {D}ain inequality with charge.
\newblock {\em Journal of Physics A: Mathematical and Theoretical},
  43(28):285202, 2010.

\bibitem{Dain05d}
Sergio Dain.
\newblock Proof of the (local) angular momemtum-mass inequality for
  axisymmetric black holes.
\newblock {\em Class. Quantum. Grav.}, 23:6845--6855, 2006.

\bibitem{Dain06c}
Sergio Dain.
\newblock Proof of the angular momentum-mass inequality for axisymmetric black
  holes.
\newblock {\em J. Differential Geometry}, 79(1):33--67, 2008.

\bibitem{dain10d}
Sergio Dain.
\newblock Extreme throat initial data set and horizon area-angular momentum
  inequality for axisymmetric black holes.
\newblock {\em Phys. Rev. D}, 82(10):104010, Nov 2010.

\bibitem{Dain:2010uh}
Sergio Dain and Maria E.~Gabach Clement.
\newblock {Small deformations of extreme Kerr black hole initial data}, 2010.

\bibitem{Dain:2008yu}
Sergio Dain and Mar\'ia~Eugenia Gabach~Cl\'ement.
\newblock {Extreme Bowen-York initial data}.
\newblock {\em Class. Quantum. Grav.}, 26:035020, 2009.

\bibitem{Dain02c}
Sergio Dain, Carlos~O. Lousto, and Ryoji Takahashi.
\newblock New conformally flat initial data for spinning black holes.
\newblock {\em Phys. Rev. D}, 65(10):104038, 2002.

\bibitem{Dain:2008ck}
Sergio Dain, Carlos~O. Lousto, and Yosef Zlochower.
\newblock {Extra-Large Remnant Recoil Velocities and Spins from Near-
  Extremal-Bowen-York-Spin Black-Hole Binaries}.
\newblock {\em Phys. Rev. D}, 78:024039, 2008.

\bibitem{Gibbons72}
G.~W. Gibbons.
\newblock The time symmetric initial value problem for black holes.
\newblock {\em Commun. Math. Phys.}, 27:87--102, 1972.

\bibitem{Hildebrandt77}
St{\'e}fan Hildebrandt, Helmut Kaul, and Kjell-Ove Widman.
\newblock An existence theorem for harmonic mappings of {R}iemannian manifolds.
\newblock {\em Acta Math.}, 138(1-2):1--16, 1977.

\bibitem{Li92}
Yan~Yan Li and Gang Tian.
\newblock Regularity of harmonic maps with prescribed singularities.
\newblock {\em Commun. Math. Phys.}, 149(1):1--30, 1992.

\bibitem{Misner63}
C.~W. Misner.
\newblock The method of images in geometrostatics.
\newblock {\em Ann. Phys. (N.Y.)}, 24:102--117, 1963.

\end{thebibliography}

\end{document}